\documentclass{article}
\usepackage{graphicx} 
\usepackage[hyphens]{url}
\usepackage{amsmath}
\usepackage{amssymb}
\usepackage{amsthm}
\usepackage{xcolor,hyperref}
\usepackage{algpseudocode}
\usepackage{algorithm}
\usepackage{booktabs}
\usepackage{comment}
\usepackage[normalem]{ulem}
\usepackage{authblk}

\usepackage{tikz}
\usetikzlibrary{arrows,positioning, calc}
\tikzstyle{vertex}=[draw,fill=black!15,circle,minimum size=20pt,inner sep=0pt]

\newtheorem{proposition}{Proposition}
\newtheorem{corollary}{Corollary}
\newtheorem{conjecture}{Conjecture}
\newtheorem{example}{Example}
\newtheorem{remark}{Remark}

\DeclareMathOperator{\res}{Res}

\def\degDistr{\delta}

\title{Attacking the {\it Polynomials in the Maze of Finite Fields} problem\footnote{This work was supported in part by the Research Council of Finland grant (\#351271), the Business Finland Co-Innovation Consortium grant (\#BFRK/473/31/2024, 5845483), and the R\&D+i project PID2021-124613OB-I00 funded by MICIU/AEI/10.13039/501100011033 and FEDER, EU.}}

\author[1]{\'Angela Barbero}
\author[2]{Ragnar Freij-Hollanti}
\author[2]{Camilla Hollanti}
\author[3]{Håvard Raddum}
\author[4]{Øyvind Ytrehus}
\author[3,4]{Morten Øygarden}
\affil[1]{Universidad de Valladolid, Spain}
\affil[2]{Aalto University, Finland}
\affil[3]{Simula UiB, Norway}
\affil[4]{University of Bergen, Norway}

\date{}

\begin{document}

\maketitle

\begin{abstract} In April 2025 GMV announced a competition for finding the best method to solve a particular polynomial system over a finite field.  In this paper we provide a method for solving the given equation system significantly faster than what is possible by brute-force or standard Gröbner basis approaches.  The method exploits the structured sparsity of the polynomial system to compute a univariate polynomial in the associated ideal through successive computations of \emph{resultants}. A solution to the system can then be efficiently recovered from this univariate polynomial.
Pseudocode is given for the proposed {\tt ResultantSolver} algorithm, along with experiments and comparisons to rival methods. We also discuss further potential improvements, such as parallelizing parts of the computations.
\end{abstract}

\section{Background}
In April 2025 the company GMV issued a competition for finding good methods to solve a polynomial system of equations of a particular structure.  The system is defined over a finite field $\mathbb{F}_p$, and the number of variables in the system is linked to the bit-size of $p$.  GMV were well aware of standard methods for solving non-linear equation systems over finite fields, but were looking for techniques that could exploit the special structure of the system.  More than 100 participants took part in the competition, and the winners were announced in July 2025 \cite{win}.  In this paper we present the method we developed for tackling the GMV system.

Our strategy for solving the system of polynomial equations given in \cite{GMVproblem} is to take advantage of the particular structure of the system, which lends itself to the theory of resultants.  The idea is inspired by recent work in cryptanalysis where resultants have been used to solve the constrained input/constrained output (CICO) problem for some permutations over large finite fields more efficiently than using general solving methods \cite{res1,res2}.  In particular, the CICO problem in this setting also boils down to solving a structured multivariate polynomial system over a large prime field.

A resultant takes two polynomials $f$ and $g$ and a target variable $x_i$ as input and produces a new polynomial where the target variable has been eliminated: 
\[
\res(f(x_1,\ldots,x_n),g(x_1,\ldots,x_n);x_i)=h(x_1,\ldots,x_{i-1},x_{i+1},\ldots,x_n).
\]
One useful property is that the new polynomial will preserve the common roots of $f$ and $g$.  The resultant is computed as the determinant of the Sylvester matrix of $f$ and $g$ and is constructed as follows.  Let $\deg(f)=d$ and $\deg(g)=e$ and write $f$ and $g$ as 
\begin{align*}
    f = \alpha_dx_i^d + \alpha_{d-1}x_i^{d-1} + \ldots + \alpha_1x_i + \alpha_0\\
    g = \beta_ex_i^e + \beta_{e-1}x_i^{e-1} + \ldots + \beta_1x_i + \beta_0,\\
\end{align*}
where the $\alpha_j$ and $\beta_j$ are in $\mathbb{F}_p[x_1,\ldots,x_{i-1},x_{i+1},\ldots,x_n]$.  The resultant of $f$ and $g$ is computed as the $(d+e)\times (d+e)$ determinant
\begin{equation}
\res(f,g;x_i)=\left|\begin{array}{ccccccc}
     \alpha_d & \alpha_{d-1} & \cdots & \alpha_0 & 0 &  \cdots & 0\\
     0 & \alpha_d & \alpha_{d-1} & \cdots & \alpha_0  &  \cdots & 0\\     
     \ddots & \ddots & & & \ddots & \ddots\\
     0 & \cdots & 0 & \alpha_d & \alpha_{d-1} & \cdots & \alpha_0\\\
     \beta_e & \beta_{e-1} & \cdots & \beta_0 & 0 &  \cdots & 0\\
     0 & \beta_e & \beta_{e-1} & \cdots & \beta_0  &  \cdots & 0\\     
     \ddots & \ddots & & & \ddots & \ddots\\
     0 & \cdots & 0 & \beta_e & \beta_{e-1} & \cdots & \beta_0
\end{array}\right|.
\end{equation}

\section{The GMV Polynomial System}\label{sec:GMVPolys}
The system given in the GMV challenge is defined over $\mathbb{F}_p[x_0,\ldots,x_n]$ as follows:
\begin{equation}
\begin{array}{rl}
     f_0 = & b_1x_0 + x_1 + x_n = 0 \\
     f_1' = & a_0x_0^3 - (b_0x_0 + 2x_1)(a_1x_0^2 + x_1^2 + 1) = 0\\
     f_i = & x_i(x_{i-1}^2 - 1) + (a_ix_1 + b_ix_n)(2x_ix_{i-1} - x_{i-1}^2 + 1) + 2x_{i-1} = 0, \\
        & \mbox{for }2\leq i\leq n-1\\
     f_n = & t(x_{n-1}^2 + 2x_nx_{n-1} - 1) - x_n(x_{n-1}^2 - 1) + 2x_{n-1} = 0,
\end{array}  \label{eq:GMV}
\end{equation}
for constants $t,a_j,b_j\in \mathbb{F}_p$. The exact manner in which these constants are chosen is not given in \cite{GMVproblem}. We will instead, for the remainder of this document, assume that the constants are generic, and will in self-generated experiments sample them uniformly at random from $\mathbb{F}_p$.

Observe that the above polynomial system is particularly suitable for the use of resultants. Indeed, for $i=2,\ldots,n-1$ we have that the variable $x_i$ only appears in the two polynomials $f_i$ and $f_{i+1}$, and can hence be eliminated from the system by a single resultant computation. More formally, for the ideal $I = \langle f_0,f_1',f_2,\ldots f_n\rangle$, we have 
\begin{align*}
    I\cap \mathbb{F}_p[x_1,\ldots,x_{i-1},x_{i+1},\ldots x_n] = \langle f_0,f_1',f_2,\ldots,f_{i-1},g_i,f_{i+2}\ldots f_n\rangle,
\end{align*}
where $g_i = Res(f_i,f_{i+1};x_i)$. This process can now be continued using $x_{i-1}$ and the polynomials $f_{i-1}$ and $g_i$, and so on.  One item of notation used in the rest of this paper: Let the degree of a variable $x_i$ in any polynomial $f\in\mathbb{F}[x_1,\ldots,x_n]$ be denoted $\deg_{x_i}(f)$. 

\subsection{Initial Preparation}
\label{sec:initial}
We initialize and modify the given system slightly as follows.  First, use\\
$b_1x_0+x_1+x_n=0$ to eliminate $x_0$ from $f_1'$ and call the resulting polynomial for $f_1$:
\[
f_1(x_1,x_n) = f_1'(-b_1^{-1}(x_1+x_n),x_1).
\]
The polynomial $f_1$ then contains 
$$(-a_0/b_1^3 + a_1b_0/b_1^3 - 2a_1/b_1^2 + b_0/b_1-2)x_1^3$$
as a degree-3 term in $x_1$ and all other terms in $f_1$ have degree $\leq 2$ in $x_1$. In principle the coefficient above may evaluate to zero, but with random coefficients from a large field this is an unlikely event.
  
From $f_1 = 0$ we single out the pure $x_1^3$-term as $x_1^3 = r(x_1,x_n)$, where the $\deg_{x_1}(r)\leq 2$.  We can then perform all computations not involving $f_1$ in the ring $\mathbb{F}_p[x_1,\ldots,x_n]/\langle x_1^3 - r(x_1,x_n)\rangle$, ensuring that the degree of $x_1$ is always smaller than $3$ in the output of any computation.

\subsection{High-level Algorithm}
The high-level algorithm to solve the system is to eliminate the variables\\
$x_{n-1}, x_{n-2},\ldots,x_2$ using one resultant computation to eliminate each variable.  In the end we will be left with $f_1(x_1,x_n)$ and $g_2(x_1,x_n)$, where $g_2(x_1,x_n)$ is the final output of $n-2$ resultant computations.  We then do one final resultant computation $\res(f_1(x_1,x_n),g_2(x_1,x_n);x_1)=u(x_n)$ to give us a univariate polynomial $u(x_n)$ that is in the ideal $\langle f_1,f_2,\ldots,f_n\rangle$. The roots of $u(x_n)$ can now be found with an expected time complexity of $O(n\log^2 n\log pn)$ multiplications in $\mathbb{F}_p$ \cite[Corollary 14.16 \& Chapter 8]{ModernComputerAlgebra}. As already noted in \cite{GMVproblem}, it is straightforward to recover the full solution once a correct value for $x_n$ is known. 

\subsection{Example Case}
We illustrate how the basic solving algorithm works on the example given in the challenge, where $p=1523$ and $n=7$.  To better illustrate the procedure and what happens during computation we do not write out the polynomials $f_1,\ldots,f_7$, but represent them via their vector of degrees $\degDistr_i$ in each single variable.  That is, $\degDistr_i$ is given as 
\[
\degDistr_i = (\deg_{x_1}(f_i),\deg_{x_2}(f_i),\ldots,\deg_{x_7}(f_i)).
\]

The system after initialization but before any resultants have been computed is then given as
\[
\begin{array}{c|c}
    \mathrm{polynomial} & \degDistr_i=(\deg_{x_1},\ldots,\deg_{x_7})  \\
    \hline
     f_1 & (3,0,0,0,0,0,3)\\
     f_2 & (3,1,0,0,0,0,1)\\
     f_3 & (1,2,1,0,0,0,1)\\
     f_4 & (1,0,2,1,0,0,1)\\
     f_5 & (1,0,0,2,1,0,1)\\
     f_6 & (1,0,0,0,2,1,1)\\
     f_7 & (0,0,0,0,0,2,1)
\end{array}
\]

We start by taking the resultant of $f_6$ and $f_7$, eliminating $x_6$:
\[
g_6(x_1,x_5,x_7)=\res(f_6,f_7;x_6).
\]
Computing this resultant is done by computing a $3\times 3$ determinant.   Write $f_7$ as $f_7=\alpha_2x_6^2+\alpha_1x_6+\alpha_0$ where each $\alpha_i\in\mathbb{F}_p[x_7]$, and $f_6$ as $f_6=\beta_1x_6+\beta_0$, where $\beta_i\in\mathbb{F}_p[x_1,x_5,x_7]$.  Then
\[
g_6(x_1,x_5,x_7)=\left|\begin{array}{ccc}
    \beta_1 & \beta_0 & 0 \\
     0 & \beta_1 & \beta_0\\
     \alpha_2 & \alpha_1 & \alpha_0
\end{array}\right|=\alpha_2\beta_0^2 - \alpha_1\beta_1\beta_0+\alpha_0\beta_1^2.
\]
After this, the system looks like
\[
\begin{array}{c|c}
    pol & (\deg_{x_1},\ldots,\deg_{x_7})  \\
    \hline
     f_1 & (3,0,0,0,0,0,3)\\
     f_2 & (3,1,0,0,0,0,1)\\
     f_3 & (1,2,1,0,0,0,1)\\
     f_4 & (1,0,2,1,0,0,1)\\
     f_5 & (1,0,0,2,1,0,1)\\
     g_6 & (2,0,0,0,4,0,3)\\
\end{array}
\]
where we have reduced all polynomial multiplications modulo $x_1^3-r(x_1,x_n)$ to keep $\deg_{x_1}(g_6)\leq 2$.  We see that $\deg_{x_5}(g_6)=4$, the double of $\deg_{x_5}(f_6)$.  This is a general behavior of the system that we will prove below, and is the reason we end up with an overall complexity of the solving algorithm of $\mathcal{O}(2^n)$.  Likewise, we also have $\deg_{x_7}(g_6)=\deg_{x_5}(g_6)-1$.  We will also prove this relation holds in general. 

We continue by recursively computing $g_5=\res(f_5,g_6;x_5)$, down to $g_2=\res(f_2,g_3;x_2)$.  These determinants are no longer $3\times 3$ determinants, but they are sparse since $\deg_{x_i}(f_i)=1$ when $x_i$ is to be eliminated.  For each resultant, we have\footnote{The polynomials $\beta_0$, $\beta_1$ and $\alpha_j$ depend on $i$, but in order to keep the notation simpler we omit $i$ when there is no danger of confusion.} $f_i=\beta_1x_i+\beta_0$ and $g_{i+1}=\alpha_{2^{7-i}}x_i^{2^{7-i}}+\alpha_{2^{7-1}-1}x_i^{2^{7-i}-1}+\ldots+\alpha_1x_i+\alpha_0$, where $\alpha_j,\beta_j\in\mathbb{F}_p[x_1,x_{i-1},x_7]$. 
Then the $(2^{7-i}+1)\times (2^{7-i}+1)$ determinant to be computed is
\begin{equation}\label{eq:sparseRes}
\res(f_i,g_{i+1};x_i)=\left|\begin{array}{cccccc}
    \beta_1 & \beta_0 & 0 & 0 & \cdots & 0\\
     0 & \beta_1 & \beta_0 & 0 & \cdots & 0\\
     \vdots & \vdots & &\ddots & & \vdots\\
     0 & 0 & \cdots & 0 & \beta_1 & \beta_0\\\
     \alpha_{2^{7-i}} & \alpha_{2^{7-i}-1} & \alpha_{2^{7-i}-2} & \cdots & \alpha_1 & \alpha_0
\end{array}\right|.
\end{equation}
We prove later that these resultants can be efficiently computed due to their sparsity.  After successively computing the resultants $g_i=\res(f_i,g_{i+1};x_i)$ for $i=5,4,3,2$ we are left with the following system, where $x_1^3=r$ has been continuously used to keep $\deg_{x_1}\leq 2$.
\[
\begin{array}{c|l}
    pol & (\deg_{x_1},\ldots,\deg_{x_7})  \\
    \hline
     f_1 & (3,0,0,0,0,0,3)\\
     g_2 & (2,0,0,0,0,0,127)\\
\end{array}
\]
The degree $\deg_{x_7}(g_2)$ becomes $127$ and not $63$ since $f_2$ is a special case with $\deg_{x_1}(f_2)=3$, where the $f_i$ for $i>2$ have $\deg_{x_1}(f_i)=2$.  The reductions with $r(x_1,x_7)$ will therefore double the degree in $x_7$ again, to keep $\deg_{x_1}(g_2)\leq2$.  We now need to do one final resultant between $f_1$ and $g_2$ to eliminate $x_1$ and find $u(x_7)$.  This resultant is not of the sparse kind in (\ref{eq:sparseRes}), but it will only be a $5\times 5$ determinant since $\deg_{x_1}(f_1)=3$ and $\deg_{x_1}(g_2)=2$.  Computing this final resultant we then find the, in this case, unique root $x_7=180$ of $u(x_7)$. 

Note that the order in which we compute the resultants matters a lot for the complexity. This is most easily demonstrated by looking at the resulting tree structure, as we can see in Fig.  \ref{fig:Unbalanced} and Fig. \ref{fig:Balanced}.

\section{Behavior of Resultant Solving Method on General Systems}
\label{sec:BSMGS}

In this section we show and prove some useful properties of the resultants that apply to the particular system in (\ref{eq:GMV}).  This section ends with showing that the univariate polynomial $u(x_n)$ in the ideal $\langle f_1,\ldots,f_n\rangle$ has degree $3(2^n-1)$.

\begin{proposition}\label{prop:sparse}

The general determinant  like the one in (\ref{eq:sparseRes}) is given by 
\begin{equation}
      \det(\alpha,\beta;h):=\left|\begin{array}{ccccc}
    \beta_1 & \beta_0 & 0 &  \cdots & 0\\
     \ddots & \ddots & & & \ddots\\
     0 & 0 & \cdots & \beta_1 & \beta_0\\\
     \alpha_{h} & \alpha_{h-1} & \cdots & \alpha_1 & \alpha_0
\end{array}\right|= \sum_{i=0}^h (-1)^i \alpha_i\beta_0^i\beta_1^{h-i}.
        \label{eq:computeRes}
    \end{equation} 
\end{proposition}

{\it Proof:} We prove the statement by induction on $h$.  For $h=1$ the formula is clearly true since
\[
\left|\begin{array}{cc}
     \beta_1 & \beta_0 \\
     \alpha_1 & \alpha_0
\end{array}\right|=\alpha_0\beta_1 - \alpha_1\beta_0 = \sum_{i=0}^1 (-1)^i \alpha_i\beta_0^i\beta_1^{1-i}.
\]
Assume $\det(\alpha,\beta;h-1)=\sum_{i=0}^{h-1} (-1)^i \alpha_i\beta_0^i\beta_1^{h-1-i}$ (induction hypothesis). Then $\det(\alpha,\beta;h)$ can be computed using co-factor expansion along the leftmost column as

\begin{eqnarray*}
\det(\alpha,\beta;h)&=&\beta_1\cdot\left|\begin{array}{ccccc}
    \beta_1 & \beta_0 & 0 &  \cdots & 0\\
     \ddots & \ddots & & & \ddots\\
     0 & 0 & \cdots & \beta_1 & \beta_0\\\
     \alpha_{h-1} & \alpha_{h-2} & \cdots & \alpha_1 & \alpha_0
\end{array}\right|\\
&+& (-1)^h\alpha_h\cdot \left|\begin{array}{ccccc}
    \beta_0 & 0 & 0 &  \cdots & 0\\
    \beta_1 & \beta_0 & 0 & \cdots & 0\\
     \ddots & \ddots & & & \ddots\\
     0 & 0 & \cdots & \beta_1 & \beta_0
\end{array}\right|
\end{eqnarray*}

The first determinant is $\det(\alpha,\beta;h-1)$, while the second determinant is equal to $\beta_0^h$ since it is a diagonal $(h\times h)$ matrix with $\beta_0$ on the diagonal.  Using the induction hypothesis the expression therefore resolves to 

\begin{eqnarray*}
\beta_1\sum_{i=0}^{h-1}  (-1)^i \alpha_i\beta_0^i\beta_1^{h-1-i} + (-1)^h \alpha_h\beta_0^h &=& 
\sum_{i=0}^{h-1}  (-1)^i \alpha_i\beta_0^i\beta_1^{h-i} + (-1)^h \alpha_h\beta_0^h\\ &=& \sum_{i=0}^{h} (-1)^i \alpha_i\beta_0^i\beta_1^{h-i}
\end{eqnarray*}
as desired. \qed

\begin{corollary}
The resultant $g_i=\res(f_i,g_{i+1};x_i)$ can be computed with $4h$ polynomial multiplications, where $h$ is the degree of $x_i$ in the polynomial $g_{i+1}$ 
\end{corollary}

{\it Proof:} The sequences of polynomials $\beta_0^2,\ldots,\beta_0^h$ and $\beta_1^2,\ldots,\beta_1^h$ can be computed with $h-1$ multiplications each.  Each term in the sum in Proposition \ref{prop:sparse} can then be computed with $2$ multiplications, so the whole sum can be computed with $2(h+1)$ multiplications.  In total the computation of $\res(f_i,g_{i+1};x_i)$ takes $2(h-1) + 2(h+1) = 4h$ multiplications. \qed

\begin{proposition}\label{prop:elimdeg}
Let $a\geq 3$ and assume the variables $x_{n-1},\ldots,x_a$ have been eliminated using resultants, producing the polynomials $g_n=f_n,g_{n-1},\ldots,g_{a+1}$, where $g_i=\res(f_i,g_{i+1};x_i)$ for $3\leq i\leq a$.  Then $\deg_{x_{a-1}}(g_a)=2^{n-a+1}$ and $\deg_{x_n}(g_a)=2^{n-a+1}-1$.

\end{proposition}

{\it Proof:} We prove the statement by induction on $a$.  Induction basis: for $a=n$ (that is, no variables have been eliminated yet), the statement is true since the system is constructed to have $\deg_{x_{n-1}}(f_n)=2=2^{n-n+1}$ and $\deg_{x_n}(f_n)=1=2^{n-n+1}-1$.

Induction step: assume the variables $x_{n-1},\ldots,x_{a+1}$ have been eliminated and $\deg_{x_a}(g_{a+1})=2^{n-a}$ and that $\deg_{x_n}(g_{a+1})=2^{n-a}-1$.  Since $\deg_{x_a}(g_{a+1})=2^{n-a}$ the resultant computation can be done as in Proposition \ref{prop:sparse} with $h=2^{n-a}$.  Furthermore, $\deg_{x_{a-1}}(g_{a+1})=0$, since the variable $x_{a-1}$ has not been involved in any resultant eliminating $x_{n-1},\ldots,x_{a+1}$.  By construction of the initial system, we see that when writing $f_a$ as $f_a=\beta_1x_a+\beta_0$ we have $\deg_{x_{a-1}}(\beta_i)=2$ and $\deg_{x_n}(\beta_i)=1$ for $i=1,2$.  We calculate the degree of each term in $x_{a-1}$ and $x_n$ in the sum in Proposition \ref{prop:sparse}:
\begin{align*}
\deg_{x_{a-1}}(\alpha_i\beta_0^i\beta_1^{2^{n-a}-i}) & = 0+2i+2(2^{n-a}-i)= 2^{n-a+1}\\
\deg_{x_n}(\alpha_i\beta_0^i\beta_1^{2^{n-a}-i}) & = (2^{n-a}-1)+i+ (2^{n-a}-i)=2^{n-a+1}-1.
\end{align*}
This proves the claim by induction. \qed

We can successively eliminate one variable at the time starting from $x_{n-1}$ and ending with $x_3$, by computing $g_i=\res(f_i,g_{i+1};x_i)$ for $i=n-1, \ldots, 3$.  After $x_3$ has been eliminated we end up with a system of three polynomials $f_1,f_2,g_3$ in the variables $x_1,x_2,x_n$.  By Proposition \ref{prop:elimdeg} this system has the following degrees in each variable:
\begin{equation}\label{tab:3sys}
\begin{array}{c|c}
    pol & (\deg_{x_1},\ldots,\deg_{x_n})  \\
    \hline
     f_1 & (3,0,0,\ldots,0,3)\\
     f_2 & (3,1,0,\ldots,0,1)\\
     g_3 & (2,2^{n-2},0,\ldots,0,2^{n-2}-1).
\end{array}
\end{equation}
The degree of the next resultant cannot be computed from Proposition \ref{prop:elimdeg} since $\deg_{x_1}(f_2)=3$.

\begin{remark}\label{rem1}
Let us define $\deg_{x_1,x_n}(f)$ as the joint degree of $f$ in variables $x_1$ and $x_n$, that is to say, the total degree of $f$ if we substitute all variables $x_i$, $i\neq 1,n$ by a constant.  For instance, $\deg_{x_1,x_n}(x_1^kx_n^k)=2k$ but $\deg_{x_1,x_n}((x_1+x_n)^k)=k$, even when in both cases the degrees in $x_1$ and in $x_n$ are both $k$.  In a similar way as it has been done in Proposition \ref{prop:elimdeg}, by induction and careful computation of the degrees of each of the $\alpha_i$ and $\beta_i$ terms, specially taking into account that $\deg_{x_1,x_n}\beta_i=1$ for all the $\beta$ terms in all the $f_a$ for  $a=n-1,\ldots, 3$,   we can deduce that $\deg_{x_1,x_n}(g_a)=2^{n-a+1}-1$. 
\end{remark}

\begin{proposition}\label{g2}
Let $g_2=\res(f_2,g_3;x_2)$ where the degree-profile is like in (\ref{tab:3sys}).  Then $\deg_{x_n}(g_2)=2^n-1$. 
\end{proposition}

{\it Proof:} 
As mentioned, this case is a bit different to the previous ones because $\deg_{x_1}(f_2)=3$, but we  will split, as before, 
\[f_2=\beta_1x_2+\beta_0,\;\; {\rm and}\;\; g_{3}=\displaystyle\sum_{i=0}^{2^{n-2}}\alpha_i x_2^i\] 
with $\alpha_i, \beta_j\in\mathbb{F}_p[x_1,x_n]$.  A careful study shows that $\deg_{x_1}(\beta_1)=2$ and $\deg_{x_n}(\beta_1)=1$ while $\deg_{x_1, x_n}(\beta_1)=2$.  Also  $\deg_{x_1}(\beta_0)=3$ and $ \deg_{x_n}(\beta_0)=1$ while $\deg_{x_1, x_n}(\beta_0)=3$.  Finally, using the reasoning explained in Remark \ref{rem1}, $\deg_{x_1, x_n}(\alpha_i)\leq 2^{n-2}$, for all $i=0\ldots, 2^{n-2}$.

Therefore 
\[\deg_{x_1,x_n}(\alpha_i\beta_0^i\beta_1^{2^{n-2}-i})\leq 2^{n-2}-1+3i+2(2^{n-2}-i)=2^{n-2}+2^{n-1}+i-1.\]
Since $i\leq 2^{n-2}$, we have
\[\deg_{x_1,x_n}(\alpha_i\beta_0^i\beta_1^{2^{n-2}-i})\leq 2^{n-2}+2^{n-1}+2^{n-2}-1=2^n-1.\]

This shows that the combined degree of $g_2$ in variables $x_1$ and $x_n$ is at most $2^n-1$. If we then apply the reduction polynomial $r(x_1, x_n)$ to reduce the degree in $x_1$, the combined degree in $x_1$ and $x_n$ will not change because $\deg_{x_1,x_n}(r(x_1,x_n))=3=\deg_{x_1,x_n}(x_1^3)$ . 

After the reduction we still have $\deg_{x_1, x_n}(g_2)\leq 2^n-1$ and this implies $\deg_{x_n}(g_2)\leq 2^n-1$ and the proof is complete. 

\qed

\medskip

The final resultant to compute before finding the univariate polynomial in $\langle f_1,\ldots,f_n\rangle$ is $\res(f_1,g_2;x_1)$.  This resultant is not of the sparse kind since $\deg_{x_1}(f_1)=3$ and $\deg_{x_1}(g_2)=2$, while Proposition \ref{prop:sparse} only applies when one of the polynomials is linear in the variable to be eliminated.  However, since the $x_1$-degree for $f_1$ and $g_2$ is only $2$ and $3$, this last resultant will only be the determinant of a $5\times 5$ matrix.  Let $f_1=\beta_3x_1^3+\beta_2x_1^2+\beta_1x_1+\beta_0$ and $g_2=\alpha_2x_1^2+\alpha_1x_1+\alpha_0$.  Then the determinant has the following form:
\begin{equation}\label{eq:5x5}
u(x_n)=\res(f_1,g_2;x_1)=\left|\begin{array}{ccccc}
    \beta_3 & \beta_2 & \beta_1 & \beta_0 & 0\\
    0 & \beta_3 & \beta_2 & \beta_1 & \beta_0\\
    \alpha_2 & \alpha_1 & \alpha_0 & 0 & 0\\
    0 & \alpha_2 & \alpha_1 & \alpha_0 & 0\\
    0 & 0 & \alpha_2 & \alpha_1 & \alpha_0
\end{array}\right|
\end{equation}
By inspection of the decomposition of $f_1$, we find that $\deg_{x_n}(\beta_i)=3-i$.  Likewise, we have proved in Proposition \ref{g2} that the total degree of the highest-degree monomials in $g_2$ is $2^n-1$.  That is, $\deg(\alpha_2x_1^2)=\deg(\alpha_1x_1)=\deg(\alpha_0)=2^n-1$.  This means that $\deg_{x_n}(\alpha_j)=2^n-1-j$.  Using these degrees we are then able to prove the following proposition.

\begin{proposition}\label{prop:degu}
    The degree of $u(x_n)$ in (\ref{eq:5x5}) is $3\cdot(2^n-1)$.
\end{proposition}

{\it Proof:} The determinant in (\ref{eq:5x5}) is computed as the sum of products of all entries where exactly one entry is chosen from each row and column.  That is, each product that does not involve any of the $0$'s in the matrix will have the form $\beta_{i_1}\beta_{i_2}\alpha_{j_1}\alpha_{j_2}\alpha_{j_3}$.  By inspection, or tedious calculations, we can see that an invariant that holds for all possible non-zero products is $i_1+i_2+j_1+j_2+j_3=6$.  Using the the known degrees $\deg_{x_n}(\beta_i)=3-i$ and $\deg_{x_n}(\alpha_j)=2^n-1-j$ we find that the degree of each non-zero term, and hence the degree of $u(x_n)$, is 
\begin{align*}
(3-i_1)+(3-i_2)+(2^n-1-j_1)+(2^n-1-j_2)+(2^n-1-j_3) =\\
3\cdot2^n+3-(i_1+i_2+j_1+j_2+j_3)=3\cdot 2^n+3-6=3\cdot(2^n-1)
\end{align*}\qed

In the very unlikely case, in which the leader coefficients of any of the polynomials involved in the computations end up being 0, the only way it could affect the process is by reducing the final degree and consequently also the overall complexity. 

Both the resultant approach and a Gröbner basis computation for solving (\ref{eq:GMV}) tries to find the univariate polynomial $u(x_n)$ in the ideal spanned by the polynomials $f_1,\ldots,f_n$.  The degree of $u$ is given by Proposition \ref{prop:degu}, showing that just writing it into memory has both time and memory complexity of order $\mathcal{O}(2^n)$.  This gives evidence to the following conjecture.

\begin{conjecture}\label{con:complexity}
The time and memory complexity of any algorithm that solves the equation system (\ref{eq:GMV}) by finding a univariate polynomial in $x_n$ in the ideal $\langle f_1,\ldots,f_n\rangle$ must have complexity at least $\mathcal{O}(2^n)$.
\end{conjecture}

\section{The {\tt ResultantSolver} Algorithm and Implementational Considerations}

In this section we consider details of the implementation of the algorithm for solving the main problem. Sections~\ref{sec:precomp}-\ref{sec:alg} present and summarize implementation details of Algorithm~\ref{alg:1} and briefly discuss its time and memory complexity. Lastly, Section~\ref{sec:compare} compares the performance of Algorithm~\ref{alg:1} with that of a state-of-the art Gröbner basis approach using Magma.

\subsection{Precomputation}\label{sec:precomp}

According to \cite{GMVproblem}, {\it``...any approach that would provide a solution $x_n$ for a given value $t$, based on some known number of solutions $x_n$  for other values of $t$ may also be extremely useful.''}. Although we do not base the solution process ``for other values of $t$'' on the specific solution set for a given $t$ value, we observe that $t$ occurs only in $f_n$. Thus, since variables are eliminated iteratively by applying equations one by one, then if we do not use $f_n$ until the very last elimination step, all the preceding elimination steps and results are valid for any value of $t$, and can be considered a joint precomputation. This is described in Algorithm~\ref{alg:1}. As shown in Table~\ref{tab:compare} in Section~\ref{tab:compare}, execution time can be split into the precomputation part (valid for all values of $t$) and the $t$-specific part. Indeed, for the cases we have considered, execution time for the final step is small compared to the precomputation.

\subsection{Resultant Computation} \label{sec:resultantorder}
Naïvely, variables $x_{n-1}, x_{n-2},\ldots,x_2$ can be eliminated iteratively in descending order. This is what is done in the algorithm described in Section 3, when computing $g_i= \res(f_i,g_{i+1};x_i)$ with backwards induction over $i$. This computation can then be described graphically with a binary tree as in Figure~\ref{fig:Unbalanced}. In Figure~\ref{fig:Unbalanced}, $g_i\in\mathbb{F}_p[x_1,x_{i-1},x_n]$ is computed as the resultant with respect to $x_i$ of $f_i\in\mathbb{F}_p[x_1,x_{i-1},x_i,x_n]$ and $g_{i+1}\in\mathbb{F}_p[x_1,x_i,x_n]$.
\begin{center}
\begin{figure}[htb]
\begin{tikzpicture}[very thick,level/.style={sibling distance=60mm/#1}]
\node [vertex] (r){$g_{2}$}
  child {
    node [vertex] (a) {$f_{2}$}
    }
  child {
      node [vertex] {$g_{3}$}
      child {
    node [vertex] (b) {$f_{3}$}
    }
      child {
      node [vertex] {$g_{4}$}
            child {
    node [vertex] (c) {$f_{4}$}
    }
            child {
      node [vertex] {$g_{5}$}
            child {
    node [vertex] (d) {$f_{5}$}
    }
            child {
      node [vertex] {$g_{6}$}
                 child {
    node [vertex] (d) {$f_{6}$}
    }
            child {
      node [vertex] {$f_{7}$}
    }
    }
    }
    }
    }
;
\end{tikzpicture}
\caption{Binary tree describing the algorithm from Section 3, with $n=7$.
}
\label{fig:Unbalanced}
\end{figure}
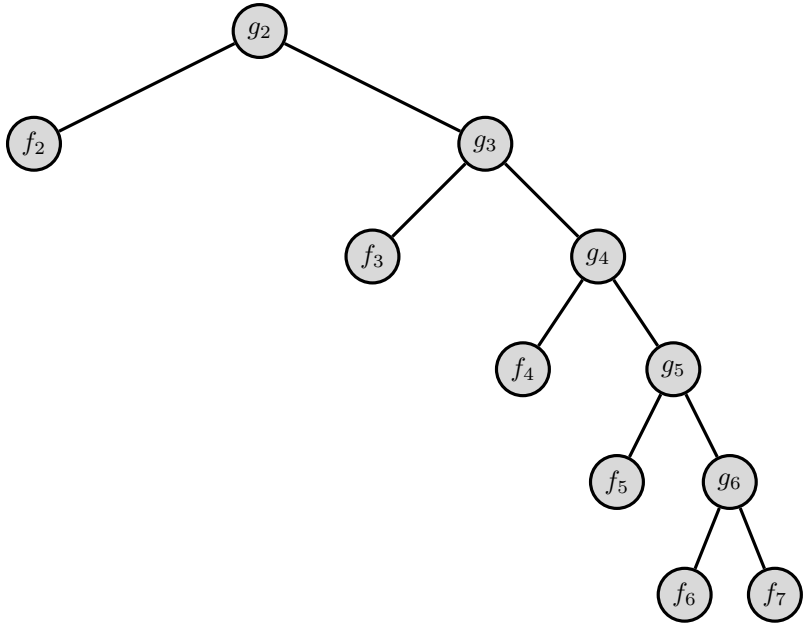
\end{center}

However, the order in which the resultants are computed is arbitrary, and a generalized version of the algorithm can be performed for an arbitrary binary tree structure on the list of polynomials $f_1, \dots f_n$. The fact that the univariate final polynomial $u(x_n)$ is the unique minimal polynomial in $\langle f_1, \dots , f_n\rangle\cap\mathbb{F}_p[x_n]$ implies that the order of the resultants does not matter for the algebra, and the propositions in Section~\ref{sec:BSMGS} can be adapted to any such order. For this reason, we generalize the construction in Section 3 by denoting $f_{[i,i]}=f_i$, and when $1\leq i\leq j\leq k< n$, introducing the polynomials $$f_{[i+1,k]}=\res(f_{[i+1,j]}, f_{[j+1,k]};x_j)\in\mathbb{F}_p[x_1, x_i,x_k,x_n].$$ In particular, we have $g_i=f_{[i,\dots n]}$ in the terminology of Section 3. À priori, this definition might depend on the $j$ chosen, but it typically will not, and in particular will not in the generic case where $f_{[i+1, k]}$ is the unique minimal polynomial in $\langle f_{i+1},\dots ,f_k\rangle\cap\mathbb{F}_p[x_1, x_i,x_k,x_n]$. Moreover, regardless of whether or not this is the case, any solution to the resultant equation will give a solution to the original equation system by backwards substitution.

As explained in Section~\ref{sec:precomp}, if we potentially want to use the same polynomial system with varying values of $t$, we want to consider the computation of $f_{[2,n-1]}$, which does not depend on $t$, as a precomputation. Moreover, in order to keep the total degrees down through the intermediate computations, it makes heuristic sense to include the polynomial $f_1$, which has degree $3$ in $x_n$, only in the last step of the calculations. The order in which the remaining resultants are combined does not matter for the number of resultants needed. However, the size of the determinants involved grows exponentially in the number of original polynomials $f_i$ combined, wherefore it is desirable to have as many computations as possible close to the leaves in the binary tree. This means that --- apart from not including $f_1$ and $f_n$ until the final step --- we want the binary tree of the computations to be as balanced as possible. 

Empirical observations of performance confirm that, for the cases we have considered, both memory usage and execution times are reduced by choosing a balanced binary tree structure for computing the resultants, as in Figures~\ref{fig:Balanced} and Figure~\ref{fig:Finalize}.  For clarity, in Figures~\ref{fig:Balanced} and ~\ref{fig:Finalize} as well as in the discussion of Algorithm~\ref{alg:1} in Section~~\ref{sec:alg}, we write out both the polynomials and their arguments $x_{U}$, which is used as shorthand for $\{x_i : i\in U\}$.  The index of the variable to be eliminated in the next step of the algorithm is underlined in each node. 

\begin{center}
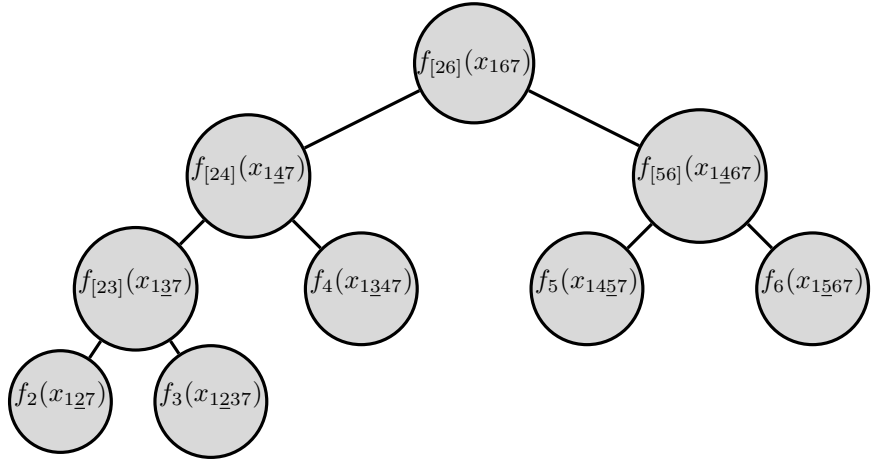
\begin{figure}
\begin{tikzpicture}[very thick,level/.style={sibling distance=60mm/#1}]
\node [vertex] (r){$f_{[26]}(x_{167})$}
        child {
    node [vertex] (a) {$f_{[24]}(x_{1\sout{4}7})$}
    child {
      node [vertex] {$f_{[23]}(x_{1\sout{3}7})$}
      child {
        node [vertex] {$f_2(x_{1\sout{2}7})$}
      }
      child {node [vertex] {$f_3(x_{1\sout{2}37})$}}
    }
    child {
      node [vertex] {$f_4(x_{1\sout{3}47})$}
    }
  }
  child {
    node [vertex] {$f_{[56]}(x_{1\sout{4}67})$}
    child {
      node [vertex] {$f_5(x_{14\sout{5}7})$}
    }
    child {
      node [vertex] {$f_6(x_{1\sout{5}67})$}
    }
  }
  ;
\end{tikzpicture}
\caption{Binary tree describing the iterative pairwise computing of the polynomials $f_{[i+1, j]}\in\mathbb{F}_p[x_1, x_{i}, x_j,x_n]$. See Example \ref{ex:n7} for explicit description for the first resultants.  The notation $f(x_{a\sout{b}c})$ means that the polynomial $f$ contains variables $x_a,x_b,x_c$ and the variable $x_b$ is eliminated in the next step.}
\label{fig:Balanced}
\end{figure}
\end{center}

\begin{figure}
\begin{center}
\begin{tikzpicture}[very thick,level/.style={sibling distance=60mm/#1}]
\node [vertex] (r){$u(x_7)$}
    child {
      node [vertex] {$f_{1}(x_{\sout{1}7})$}
      }
      child {
      node [vertex] {$g_{2}(x_{\sout{1}7})$}
        child {
        node [vertex] {$f_{[26]}(x_{\sout{1}67})$}
      }  
        child {
      node [vertex] {$f_7(x_{\sout{6}7})$}
  }
  }
  ;
\end{tikzpicture}
\end{center}
\caption{Final steps of the algorithm, when $n=7$. Notice that only this part of the algorithm depends on $f_7$, and thereby on the parameter~$t$. The value $x_n$ for a solution to the polynomial $u$ is the desired output.}
\label{fig:Finalize}
\end{figure}
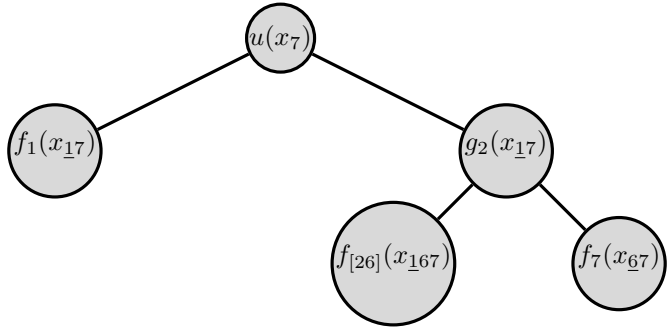

In the following example, we give the first resultant polynomials corresponding to Fig. \ref{fig:Balanced}. These resultants can be computed in parallel due to the balanced tree structure. 

\begin{example}\label{ex:n7}
System initialized: $n=7$, $p$ has $11$ bits. First pair of polynomials:
\begin{eqnarray*}
f_5 &=& 1279x_1x_4^2+488x_1x_4x_5+x_4^2x_5+1456x_4^2x_7+134x_4x_5x_7\\
&&+244x_1+2x_4+1522x_5+67x_7\\
f_6 &=& 92x_1x_5^2+1339x_1x_5x_6+x_5^2x_6+664x_5^2x_7+195x_5x_6x_7\\
&&+1431x_1+2x_5+1522x_6+859x_7
\end{eqnarray*}
Resultant polynomial of the first pair with $x_5$ eliminated:
\begin{eqnarray*}
\res(f_6,f_5;x_5) &=& 933x_1^2x_4^4x_6+388x_1^2x_4^4x_7+1494x_1^2x_4^3x_6x_7+936x_1x_4^4x_6x_7\\
&&+431x_1x_4^4x_7^2+1322x_1x_4^3x_6x_7^2+801x_4^4x_6x_7^2+49x_4^4x_7^3\\
&&+1327x_4^3x_6x_7^3+686x_1^2x_4^3+646x_1x_4^4+494x_1^2x_4^2x_6\\
&&+462x_1x_4^3x_6+1522x_4^4x_6+718x_1^2x_4^2x_7+698x_1x_4^3x_7\\
&&+679x_4^4x_7+29x_1^2x_4x_6x_7+476x_1x_4^2x_6x_7+330x_4^3x_6x_7\\
&&+460x_1x_4^2x_7^2+158x_4^3x_7^2+201x_1x_4x_6x_7^2+1286x_4^2x_6x_7^2\\
&&+1229x_4^2x_7^3+196x_4x_6x_7^3+837x_1^2x_4+693x_1x_4^2\\
&&+1519x_4^3+933x_1^2x_6+1061x_1x_4x_6+6x_4^2x_6\\
&&+388x_1^2x_7+825x_1x_4x_7+495x_4^2x_7+936x_1x_6x_7\\
&&+1193x_4x_6x_7+431x_1x_7^2+1365x_4x_7^2+801x_6x_7^2\\
&&+49x_7^3+646x_1+4x_4+1522x_6+679x_7
\end{eqnarray*}
Second pair of polynomials: 
\begin{eqnarray*}
f_3 &=& 1171x_1x_2^2+704x_1x_2x_3+x_2^2x_3+685x_2^2x_7+ 153x_2x_3x_7\\
&&+352x_1+2x_2+1522x_3+838x_7\\
f_4 &=& 1288x_1x_3^2+470x_1x_3x_4+x_3^2x_4+1409x_3^2x_7+ 228x_3x_4x_7\\
&&+ 235x_1+2x_3+1522x_4+114x_7
\end{eqnarray*}
Second resultant polynomial with $x_3$ eliminated:
\begin{eqnarray*}
\res(f_4,f_3;x_3) &=& 1497x_1^2x_2^4x_4+1003x_1^2x_2^4x_7
+557x_1^2x_2^3x_4x_7+1014x_1x_2^4x_4x_7\\
&&+365x_1x_2^4x_7^2+63x_1x_2^3x_4x_7^2+830x_2^4x_4x_7^2+1256x_2^4x_7^3\\
&&+1068x_2^3x_4x_7^3+1419x_1^2x_2^3+931x_1x_2^4+156x_1^2x_2^2x_4\\
&&+845x_1x_2^3x_4+1522x_2^4x_4+74x_1^2x_2^2x_7+1010x_1x_2^3x_7\\
&&+9x_2^4x_7+966x_1^2x_2x_4x_7+8x_1x_2^2x_4x_7+1487x_2^3x_4x_7\\
&&+856x_1x_2^2x_7^2+274x_2^3x_7^2+1460x_1x_2x_4x_7^2
+1112x_2^2x_4x_7^2\\
&&+79x_2^2x_7^3+455x_2x_4x_7^3+104x_1^2x_2
+506x_1x_2^2+1519x_2^3\\
&&+1497x_1^2x_4+678x_1x_2x_4+6x_2^2x_4
+1003x_1^2x_7\\
&&+513x_1x_2x_7+1469x_2^2x_7+1014x_1x_4x_7+36x_2x_4x_7\\
&&+365x_1x_7^2+1249x_2x_7^2+830x_4x_7^2+1256x_7^3\\
&&+931x_1+4x_2+1522x_4+9x_7
\end{eqnarray*}
After computing the rest of the pairwise resultants as shown in Fig. \ref{fig:Balanced}, we are left with variables $x_1, x_6, x_7$ in the topmost resultant $f_{[26]}$. Now, we can combine this with $f_7$ and eliminate $x_6$. Finally, we combine with $f_1$ (from which $x_0$ is already eliminated in the initial phase) to obtain a univariate polynomial $u(x_7)$.
\end{example}

\subsection{The {\tt ResultantSolver} Algorithm}
\label{sec:alg}
This section contains a pseudocode-style algorithm in two parts, Part (1) and Part (2). Part (1) computes the resultants of all polynomials $f_2,\ldots,f_{n-1}$, ending in a single polynomial $g_3$.  After this part we are left with a set of three polynomial equations $f_1=0, g_3=0$ and $f_n=0$ in the variables $x_1,x_{n-1}$, and $x_n$, and only $f_n$ contains the coefficient $t$.  Part (2) proceeds to derive the final univariate polynomial $u(x_n)$ for a specific value of $t$ by computing the last resultants $g_2=\res(f_n,g_3;x_{n-1})$ and $u(x_n)=\res(f_1,g_2;x_1)$.  Finally, a root-finding algorithm is used to solve $u(x_n)=0$.

\textcolor{red}{
\begin{algorithm}[hbt!]
\caption{{\tt ResultantSolver}: \\Pseudoalgorithm describing the procedure for solving  (\ref{eq:GMV})}
\label{alg:1}
\begin{algorithmic}
\Require The GMV equation system of equation (\ref{eq:GMV})
\Ensure Values for $x_n$ satisfying equation system.
\State {\it Part (1): Precomputation}  
\State Use $f_0$ to eliminate $x_0$, replace $f_1'$ by $f_1$ as described in Section \ref{sec:initial}.
\State $S \gets \{f_2,\ldots, f_{n-1}\}$
\For {each variable $i$ in $\{2,\ldots,n-2\}$ in the order given in \ref{sec:resultantorder}}  \do\\
    \State Choose two polynomials $f,g  \in S$ whose variable sets both contain $x_i$
    \State Replace $f$ and $g$ in $S$ by $\res(f,g;x_i)$ using (\ref{eq:computeRes}) \Comment{eliminates variable $x_i$}
\EndFor 
\State $g_3\gets$ single remaining polynomial in $S$
\State {\it Part (2):  (Can be iterated over different values of $t$)}  
\For{multiple values of $t$}
    \State Set $t$ in $f_n$
    \State $g_2\gets \res(f_n,g_3;x_{n-1})$
    \State $u\gets \res(f_1,g_2;x_1)$
    \State Find the root(s) of $u(x_n)$ in $\mathbb{F}_p$ using the Berlekamp--Rabin algorithm.
\EndFor
\end{algorithmic}
\end{algorithm}}

\subsubsection{Time Complexity} 
By Conjecture~\ref{con:complexity}, we expect that the complexity has a factor that is exponential in $n$. Also, having to work with integers modulo (a very large) $p$, we also expect to see a factor $\log_2(p)$. Thus, the overall execution time of Parts (1)+(2) can be assumed to be of the form ${\cal O}(2^{un} \log_2(p))$ for some $u>0$. 
{Experimental evidence (see Table~\ref{tab:compare}) suggests that in many cases we may have $u<1$.} 

\subsubsection{Memory Complexity} 
It is tricky to give an estimate on optimum memory usage. It is  not even clear that the last resultant ending in degree $3(2^n-1)$ is the heaviest step; in fact, the last resultant in Part (1) may actually be  bigger. Hence, while we observe that memory is a bottleneck for the algorithm, it is difficult to analyze all the options for reducing memory in a dynamic setting, also briefly discussed in Section~\ref{sec:improvements}.

\subsection{Comparison with Existing Methods}
\label{sec:compare}
In the following we compare the efficiency of our method with the current state-of-the-art in solving (generic) polynomial systems. The latter is represented by the \textrm{Variety}($\cdot$) functionality in the computer algebra system Magma V2.27-1. For the ideals associated with the GMV polynomial systems described in Section \ref{sec:GMVPolys}, this \textrm{Variety}-method consists of three steps. First, it computes a Gröbner basis in the grevlex order using the F$_4$-algorithm \cite{F4}. In the second step, this Gröbner basis is transformed to a Gröbner basis in the lexicographical order using a variant of the FGLM-algorithm \cite{FGLM}. This second basis contains a univariate polynomial in the ideal. The third and final step finds the roots of this univariate polynomial, and recovers the remainder of the associated solution. In our tests the running time has been dominated by the first two steps.

In Table \ref{tab:CompareTimes} we compare the running times of the Variety method in Magma with the resultant approach described earlier, for GMV polynomial systems with randomly generated coefficients $a_i,b_i,t$. The running times in Table~\ref{tab:CompareTimes} were  obtained by running the program on the same computer with~192GB~of~RAM. For the reference Gröbner basis algorithm, we observe that the execution times grow roughly by a factor $\approx \mathcal{O}(2^{3n})$ when $p$ is fixed. If the underlying ideal is in Shape position, then this is consistent with Proposition \ref{prop:degu}, as we expect the FGLM step to grow by a factor $\mathcal{O}(2^{\omega n})$, where practical values for the linear algebra constant lies in $2.81\leq \omega \leq 3$. 

\begin{table}
\centering
\begin{tabular}{c|cccccc}
\toprule
    $n$ & 7 & 8 & 9 & 10 & 11 & 12 \\
    \midrule 
	Magma & 0.28 & 1.95 & 13.91 & 100.33 & 781.38 & 6246.79 \\  
	\midrule
	Part 1 & $0.002$ & $0.007$ & $0.019$& $0.049$ & $0.090$ & $0.312$ \\
    \midrule
    Part 2  & $0.005$  & $0.010$  & $0.019$ & $0.032$  & $0.070$  & $0.152$ \\ 
    \bottomrule
\end{tabular}
\caption{
Execution times (in seconds) for solving the GMV polynomial systems (\ref{eq:GMV}) for increasing values of $n$. ``Magma" gives the total time cost for the Variety method implemented in Magma V2.27-1.  Part 1 gives the total time cost for part 1 of Algorithm~ \ref{alg:1} (packing the sparse resultants) and Part 2 gives the additional time needed to compute the solution for each value of $t$.  In both cases, the tests were performed on the same hardware platform, a workstation with 192GB RAM, but without using any parallelization.  All tests have been performed with randomly generated examples over the field $\mathbb{F}_p$ with $p = 8380417$.}
\label{tab:CompareTimes}
\end{table}

\begin{table}
\centering
\begin{tabular}{c|cccccc}
\toprule
    $n$ & 13 & 14 & 15 & 16 & 17 & 18 \\
    \midrule 
	Part 1 & 0.459 & 2.425 & 5.075 & 10.333 & 23.678 & 102.566 \\  
	\midrule
	Part 2 & 0.330 & 0.731 & 1.615 & 3.534 & 7.660 & 16.471 \\
    \bottomrule
\end{tabular}
\caption{
Execution times (in seconds) for solving the GMV polynomial systems (\ref{eq:GMV}) for further values of $n$.  Part 1 gives the total time cost for part 1 of Algorithm~ \ref{alg:1} and Part 2 gives the time for each iteration of part 2 of Algorithm \ref{alg:1}.  The tests were performed on a MacBook Air with 16GB RAM and an Apple M2 processor.  All tests have been performed with randomly generated examples over the field $\mathbb{F}_p$ with $p = 8380417$.}
\label{tab:ResTimes}
\end{table}

In Table \ref{tab:compare} we give some further details on the measured complexity of running Algorithm \ref{alg:1} on larger instances of (\ref{eq:GMV}).  These instances were run on a MacBook Air with 16GB of memory and an Apple M2 processor and therefore have faster run times than the workstation.  In Appendix \ref{app:printouts} we also show the complete printouts of our program, including the actual solutions of the example systems given in \cite{GMVproblem}.

Table~\ref{tab:compare} compares computation times when solving some of the example systems. 
\begin{table}[h!]
\begin{center}
\begin{tabular}{c c c c r r r}
 \hline
 GMV Example & $n$ & $\log_2(p)$ & dup & Part(1) & Part(2) &  ret\\  [0.5ex] 
 \hline
 A & 13 & 18 & 24573 & 0.453 & & \\   
 
  & & & & & 0.262 & 1.0\\
& & & & & 0.263 & 1.0\\
 \hline
 G & 17  &  22  & 393213 & 23.577 &  & \\
  & & & & & 7.530 & $\approx 81$\\
 & & & & & 7.506 & $\approx 81$\\
 \hline
 H & 17  &  73  & 393213 & 72.129 & & \\
  & & & & & 42.087 & $\approx 160$\\
 & & & & & 53.246 & $\approx 175$\\
 \hline
\end{tabular}
 \caption{Execution times for Algorithm~\ref{alg:1}, parts (1) and (2), in seconds. Here, dup = degree of final univariate polynomial.   In each case, dup is equal to $3\cdot (2^n-1)$. ret = relative execution time for parts (1)+(2), scaled so the total time for Example A is 1.0.  Observe that the total execution times for cases H and G differ by a factor close to $2$, and also that the total execution times for cases G and A differ by a factor $\approx 81$ which is $3^{17-13}$.  For more details, please see Appendix~\ref{app:printouts}.}      
 \label{tab:compare}
 \end{center}
\end{table}

\subsection{Some Potential Improvements}
\label{sec:improvements}

In order to speed up the processing, it is possible to use some tricks that we have so far not applied.
\begin{itemize}
    \item Parallelization: 
    \begin{itemize}
        \item Certain operations of part 1 can be parallelized. In particular, processes corresponding to different subtrees of Figure~\ref{fig:Balanced} can be performed in parallel.
        \item On a more fine-grained level, the resultant computation in equation (\ref{eq:computeRes}) consists of independent parts and lends itself to parallelization.
    \end{itemize}
    \item Optimizing the choice of pairs in the binary resultant tree (cf. Table \ref{fig:Balanced}) in order to avoid ``heavy'' nodes when getting closer to the top. In other words, the depths of the left-hand and the right-hand sides of the binary tree should be as close to each other as possible. 
    \item Adaptation to hardware environment: In our experiments, the amount of RAM available is the bottleneck that restricts the cases we can attack. At the cost of having to compute the same intermediate values multiple times, we may free up memory more often than we do. This is a time/memory tradeoff, and deciding the optimum strategy may be a nontrivial task.
\end{itemize}

\section{Summary}
We have presented a method for solving the GMV equation system, which is significantly faster than what is possible by brute-force or standard Gröbner basis methods. The method utilizes the theory of \emph{resultants} and the sparsity of the system, and inherently allows for parallelization (which we have not exploited at the moment).

Although the target of $n=521$ is at the moment far away, we believe there is room for improvements so that more of the example cases can be addressed. 
In particular, fully exploiting the parallelization potential and optimizing memory usage should improve the performance further.  However, we also believe that any algebraic algorithm that singles out a univariate polynomial in the ideal generated by the system must have complexity $\mathcal{O}(2^n)$.  This means having $n\approx \log_2(p)$ is out of reach if $p$ should be a $521$-bit prime.  If the relation between $n$ and $p$ can be relaxed, it is fully possible to solve the system for a small $n$ (say, $n\approx 20$), even if $p$ is a $521$-bit prime.

\bibliographystyle{abbrv}
\bibliography{bibliography}

\appendix

\section{Printouts From Solving Example Systems}
\label{app:printouts}

For reference, we include here the printouts from our implementation of Algorithm \ref{alg:1} when run on the example systems given in \cite{GMVproblem}.  These executions were done on a MacBook Air with 16GB memory and an Apple M2 processor, and not the same machine used for timings in Table \ref{tab:CompareTimes}. 

\subsection{Illustrative Toy Example}
\begin{verbatim}
System initialized: n=7, p has 11 bits
packing sparse resultants done in 0.009 seconds
Roots of univariate polynomial of degree 381 for t=1268:
x7 = 180
solution found in 0.005 seconds    
\end{verbatim}

\subsection{Case A}
\begin{verbatim}
System initialized: n=13, p has 18 bits
packing sparse resultants done in 0.453 seconds
g2 has 24573 terms
Roots of univariate polynomial of degree 24573 for t=28555:
x13 = 176007
solution found in 0.263 seconds
g2 has 24573 terms
Roots of univariate polynomial of degree 24573 for t=91766:
x13 = 98529
x13 = 100260
x13 = 39269
solution found in 0.262 seconds
g2 has 24573 terms
Roots of univariate polynomial of degree 24573 for t=99194:
x13 = 129426
x13 = 192514
solution found in 0.263 seconds
g2 has 24573 terms
Roots of univariate polynomial of degree 24573 for t=145052:
x13 = 76075
x13 = 175914
x13 = 112372
x13 = 179016
solution found in 0.263 seconds
g2 has 24573 terms
Roots of univariate polynomial of degree 24573 for t=85098:
x13 = 25091
solution found in 0.263 seconds
\end{verbatim}

\subsection{Case B}

\begin{verbatim}
System initialized: n=13, p has 18 bits
packing sparse resultants done in 0.456 seconds
g2 has 24573 terms
Roots of univariate polynomial of degree 24573 for t=2678:
x13 = 211
x13 = 19582
solution found in 0.262 seconds
g2 has 24573 terms
Roots of univariate polynomial of degree 24573 for t=142207:
x13 = 1091
solution found in 0.262 seconds
g2 has 24573 terms
Roots of univariate polynomial of degree 24573 for t=179873:
x13 = 165690
x13 = 136856
x13 = 31030
solution found in 0.262 seconds
g2 has 24573 terms
Roots of univariate polynomial of degree 24573 for t=18905:
x13 = 80081
solution found in 0.262 seconds
g2 has 24573 terms
Roots of univariate polynomial of degree 24573 for t=59939:
x13 = 123952
x13 = 64141
solution found in 0.263 seconds
\end{verbatim}

\subsection{Case C}
\begin{verbatim}
System initialized: n=13, p has 18 bits
packing sparse resultants done in 0.451 seconds
g2 has 24573 terms
Roots of univariate polynomial of degree 24573 for t=63664:
x13 = 141159
x13 = 129533
x13 = 19991
solution found in 0.263 seconds
g2 has 24573 terms
Roots of univariate polynomial of degree 24573 for t=96122:
x13 = 159594
x13 = 191455
x13 = 60057
solution found in 0.265 seconds
g2 has 24573 terms
Roots of univariate polynomial of degree 24573 for t=141701:
x13 = 141693
solution found in 0.263 seconds
g2 has 24573 terms
Roots of univariate polynomial of degree 24573 for t=35444:
x13 = 38135
x13 = 173393
x13 = 30357
solution found in 0.265 seconds
g2 has 24573 terms
Roots of univariate polynomial of degree 24573 for t=120184:
x13 = 189726
x13 = 28259
x13 = 69279
solution found in 0.264 seconds
\end{verbatim}

Solving the two example systems for $n=17$ clearly shows the solving complexity increases only linearly in $\log_2(p)$ if $n$ is kept fixed. 

\subsection{Case G}
\begin{verbatim}
System initialized: n=17, p has 22 bits
packing sparse resultants done in 23.577 seconds
g2 has 393213 terms
Roots of univariate polynomial of degree 393213 for t=3162200:
x17 = 2905037
x17 = 423964
solution found in 7.530 seconds
g2 has 393213 terms
Roots of univariate polynomial of degree 393213 for t=1948684:
x17 = 414099
solution found in 7.506 seconds
g2 has 393213 terms
Roots of univariate polynomial of degree 393213 for t=2832793:
x17 = 2773074
x17 = 344215
x17 = 2158164
solution found in 7.516 seconds
g2 has 393213 terms
Roots of univariate polynomial of degree 393213 for t=452267:
x17 = 3119706
solution found in 7.520 seconds
g2 has 393213 terms
Roots of univariate polynomial of degree 393213 for t=1114540:
x17 = 3462664
solution found in 7.501 seconds
\end{verbatim}

\subsection{Case H}
\begin{verbatim}
System initialized: n=17, p has 73 bits
packing sparse resultants done in 72.129 seconds
g2 has 393213 terms
Roots of univariate polynomial of degree 393213 for t=1058371656331684574765:
x17 = 3190850194874455142686
x17 = 4427392304825085580158
x17 = 2173747983333160095478
solution found in 42.087 seconds
g2 has 393213 terms
Roots of univariate polynomial of degree 393213 for t=21875805440386776668:
x17 = 1868566290398595250336
solution found in 53.246 seconds
g2 has 393213 terms
Roots of univariate polynomial of degree 393213 for t=3843998821665876069844:
x17 = 1298918500212281416535
x17 = 3180019097753843068048
x17 = 4714514066403504455654
solution found in 62.847 seconds
g2 has 393213 terms
Roots of univariate polynomial of degree 393213 for t=48043857385386667728:
x17 = 3085726008969165404256
solution found in 68.707 seconds
g2 has 393213 terms
Roots of univariate polynomial of degree 393213 for t=104940936778785605637:
x17 = 4179283657056012284330
x17 = 3129644745497065872124
x17 = 2702577589972463891281
x17 = 1708563171310648447784
x17 = 5180557299946937162500
solution found in 77.882 seconds

\end{verbatim}

\end{document}